\documentclass{amsart}
\usepackage{amsmath,amssymb}
\usepackage{graphicx,subfigure}
\newtheorem{theorem}{Theorem}[section]

\newtheorem{corollary}[theorem]{Corollary}

\newtheorem{lemma}[theorem]{Lemma}

\theoremstyle{definition}

\theoremstyle{remark}
\newtheorem{remark}[theorem]{Remark}

\numberwithin{equation}{section}

\newcommand{\al}{\alpha}

\newcommand{\de}{\delta}

\newcommand{\la}{\lambda}
\newcommand{\om}{\omega}

\newcommand{\te}{\theta}
\newcommand{\vp}{\varphi}

\newcommand{\De}{\Delta}

\newcommand{\Si}{\Sigma}

\newcommand{\cU}{{\mathcal U}}

\def\RR{\mathbb{R}}

\renewcommand\SS{\mathbb{S}}
\newcommand{\cM}{{\mathcal M}}

\newcommand\hk{{\widehat k}}
\newcommand{\pd}{\partial}
\newcommand\minus\backslash

\newcommand\lan\langle
\newcommand\ran\rangle

\newcommand{\rank}{\operatorname{rank}}

\newcommand{\e}{{e}}

\DeclareMathOperator\Real{Re}

\renewcommand\leq\leqslant
\renewcommand\geq\geqslant
\newlength{\intwidth}

\newcommand\loc{_{\mathrm{loc}}}

\DeclareMathOperator\Imag{Im}

\begin{document}

\title[Knotted zeros in the eigenfunctions of the harmonic
oscillator]{A problem of Berry and knotted zeros\\ in the eigenfunctions of the harmonic oscillator}

\author{Alberto Enciso}
\address{Instituto de Ciencias Matem\'aticas, Consejo Superior de
  Investigaciones Cient\'\i ficas, 28049 Madrid, Spain}
\email{aenciso@icmat.es, david.hartley@icmat.es, dperalta@icmat.es}

\author{David Hartley}
%\address{Instituto de Ciencias Matem\'aticas, Consejo Superior de
%  Investigaciones Cient\'\i ficas, 28049 Madrid, Spain}
%\email{david.hartley@icmat.es}

\author{Daniel Peralta-Salas}
%\address{Instituto de Ciencias Matem\'aticas, Consejo Superior de
%  Investigaciones Cient\'\i ficas, 28049 Madrid, Spain}
%\email{dperalta@icmat.es}

%%    General info
%\subjclass[2010]{35B38, 58J05, 58K45}
%\date{\today}
%
%\keywords{ }
%
\begin{abstract}
We prove that, given any finite link $L$ in $\RR^3$, there is a
high-energy complex-valued eigenfunction of the harmonic oscillator
such that its nodal set contains a union of connected components %Edit: `has' changed to `contains'
diffeomorphic to~$L$. This solves a problem of Berry on the existence
of knotted zeros in bound states of a quantum system.
\end{abstract}
\maketitle

\section{Introduction}

In~\cite{Be01}, Berry conjectures that there should be complex-valued eigenfunctions
of the harmonic oscillator in $\RR^3$ whose nodal set
$\psi^{-1}(0)$ has knotted connected components, and raises the question
of whether there can be eigenfunctions of a quantum system whose nodal
set has components with higher order linking, as in the case of the
Borromean rings, see Fig.~\ref{fig:1}. Furthermore, Berry remarks that it should be possible
to construct these sets so that they are {\em structurally stable}\/
in the sense that any small enough perturbation of the corresponding
eigenfunction (in the $C^k$ norm with $k\geq1$) still has connected
components in the nodal set that are diffeomorphic to the knot or link under %Edit: `that are' added before the word diffeomorphic
consideration.

\begin{figure}[t]
  \centering
  \subfigure{
\includegraphics[scale=0.1,angle=0]{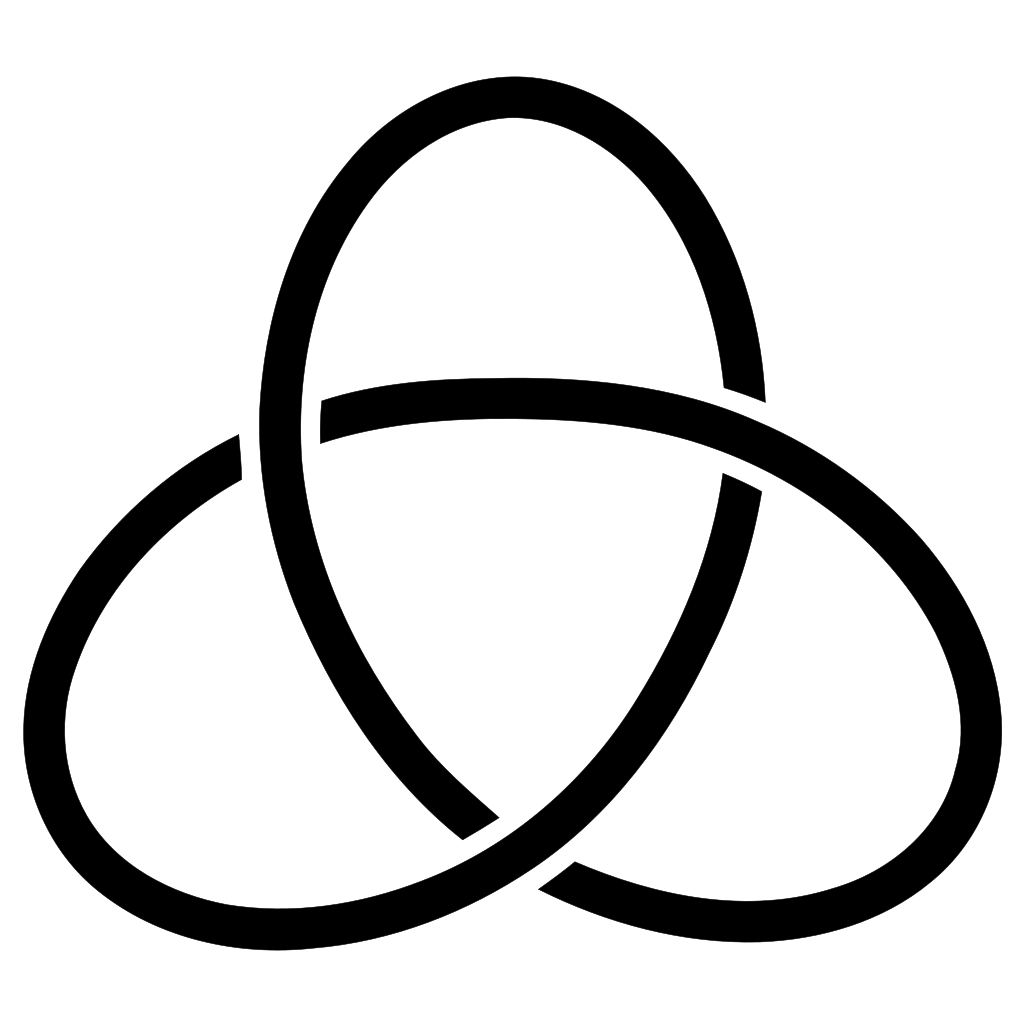}}\hspace{5em}
  \subfigure{
\includegraphics[scale=0.1]{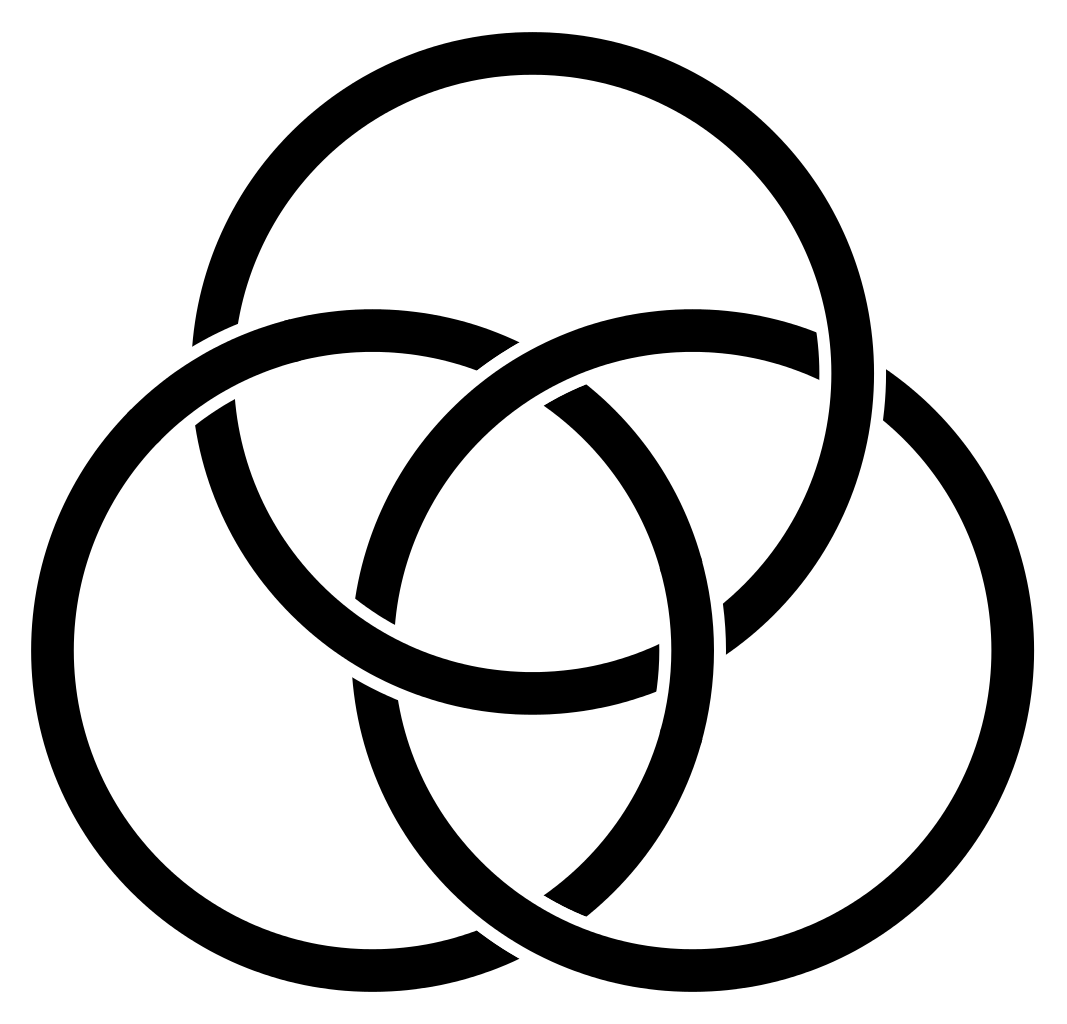}}
\caption{The problem involves showing that there are high-energy
  eigenfunctions $\psi$ of
the harmonic oscillator realizing links, e.g.\ the trefoil knot and the
Borromean rings depicted above, in their nodal set $\psi^{-1}(0)$.} \label{fig:1}
\end{figure}

As a side remark, let us recall~\cite{Be01,KL04} that a physical
motivation to study the nodal set of a quantum system is that it is
the locus of destructive interference of the wave function. It is
related to the existence of singularities (often called dislocations) of the phase
$\Imag (\log\psi)$ and of vortices in the current field $\Imag(\overline
\psi\, \nabla\psi)$. The existence of knotted structures of this type,
especially in optics and in fluid mechanics, has recently attracted
considerable attention, both from the theoretical~\cite{Annals,Acta} and
experimental~\cite{De10,Irvine} viewpoints.

The main result of this paper solves these problems of Berry by
showing that any finite link can be realized as a collection of %Edit: `finite' added before the word link
connected components of the nodal set of a high-energy eigenfunction
of the harmonic oscillator, and that the link is structurally stable
in the same sense as above. Specifically, we have the following

\begin{theorem}\label{T.main}
  Let $L$ be any finite link in $\RR^3$. Then one can deform it with
  a diffeomorphism $\Phi$ of $\RR^3$ so that $\Phi(L)$ is the union of
  connected components of the nodal set $\psi^{-1}(0)$, where $\psi$ is a
  complex-valued eigenfunction of the harmonic oscillator in
  $\RR^3$. Furthermore, the link $\Phi(L)$ is structurally stable for
  the function~$\psi$.
\end{theorem}%Edit: Second sentence structure changed: `, where $\psi$ is' after $\psi^{-1}(0)$ replacing `of' and `$\psi$' later in the sentence

We recall that the eigenfunctions of the harmonic oscillator are the
square-integrable functions $\psi$ satisfying the equation
\begin{equation}\label{harm}
-\De \psi + |x|^2 \psi=\la \psi
\end{equation}
in $\RR^3$. It is well-known that the eigenvalues are of the form 
\[
\la=2N+3\,,
\]
with $N$ a nonnegative integer, and that the degeneracy of the
corresponding eigen\-space is $\frac12(N+1)(N+2)$.

The key idea of the proof of Theorem~\ref{T.main} is that, using techniques
introduced in~\cite{Annals,Adv,Acta}, one can prove that there are %Edit: `the' changed to `there'
complex-valued solutions to the Helmholtz
equation
\[
\De \vp+\vp=0
\]
in $\RR^3$, such that the link $L$ is a union of connected components %Edit: comma added before `such that'
of the nodal set $\vp^{-1}(0)$, up to a diffeomorphism. This is pertinent to the study of the eigenvalues of the harmonic
oscillator because, in balls of radius $\la^{-1/2}$, the high-energy asymptotics of
the eigenfunctions are determined by the Helmholtz
equation (see e.g.~\cite{EJN,survey}). Heuristically, one can understand why this is true by
introducing the rescaled variable $\tilde x:=\la^{1/2}\, x$, in terms of %Edit: `-1/2' changed to `1/2'
which Eq.~\eqref{harm} is read as
\[
\De_{\tilde x}\psi +\psi = \frac{|\tilde x|^2\, \psi}{\la^2}\,.
\]
The way to make this precise is by computing the high-order
asymptotics of the Laguerre polynomials, which govern the radial
part of the eigenfunctions of the harmonic oscillator. Going over
the fine details we will see that the accidental %Edit: 'of the demonstration' deleted
degeneracy of the eigenvalues of the harmonic oscillator is an
essential ingredient of the proof too,
essentially because it ensures the existence of families of isoenergetic
eigenfunctions with a rich behavior in the angular variables.

The proof of Theorem~\ref{T.main} is given in Section~\ref{S.main},
although the proofs of two technical lemmas are relegated to
Sections~\ref{S.lemma1} and~\ref{S.lemma2}. To conclude this paper, in
Section~\ref{S.remark} we will state and discuss a 
higher-dimensional counterpart of the main theorem that can be proved
using the same argument.

\section{Proof of Theorem~\ref{T.main}}
\label{S.main}

Let us begin by fixing an orthogonal basis of eigenfunctions associated with the
harmonic oscillator Hamiltonian. Specifically, we will take
\begin{equation}\label{psiklm}
\psi_{klm}:= \e^{-\frac{r^2}2}\, r^l\, L^{l+\frac12}_k(r^2)\, Y_{lm}(\theta,\phi)\,,
\end{equation}
where $(r,\theta,\phi)$ are spherical coordinates and we are using the
standard notation for the Laguerre polynomials and the spherical
harmonics. Here the indices of the eigenfunctions range over the set
\[
k\geq0\,,\qquad l\geq0\,,\qquad -l\leq m\leq l
\]
and the eigenvalue corresponding to $\psi_{klm}$ is
\[
\la_{kl}:=4k+2l+3\,.
\]
Notice that the eigenvalue is independent of~$m$.

In the following lemma we will describe the behavior of the
eigenfunction $\psi_{klm}$ and its gradient for large values of
$k$. To state this result, we will use the notation $e_r:=x/r$ for the
unit vector in the radial direction and denote by $\nabla_{\SS^2}Y(\te,\phi)$ the
gradient (in the unit sphere) of a function $Y(\te,\phi)$ of the
angular variables. The proof of the lemma is given in Section~\ref{S.lemma1}.

\begin{lemma}\label{L.asymp}
Let us fix some integers~$l$ and~$m$ as above. Uniformly for $r\leq R$, the
eigenfunction $\psi_{klm}$ admits the asymptotic expansion
\begin{align*}
\psi_{klm}(x)&=A_{kl}\,\big[j_l(\sqrt{\la_{kl}}\, r)+ O(\tfrac1k)\big]\,
Y_{lm}(\theta,\phi)\,,\\
\nabla \psi_{klm}(x)&=\sqrt{\la_{kl}}\,
A_{kl}\,\big[j_l'(\sqrt{\la_{kl}}\, r)+ O(\tfrac1k)\big]\, Y_{lm}(\theta,\phi)\, e_r\\&+ A_{kl}\,\big[j_l(\sqrt{\la_{kl}}\, r)+O(\tfrac1k)\big]\,\frac{\nabla_{\SS^2}Y_{lm}(\theta,\phi)}{r},
\end{align*} %Edit: Added missing spherical harmonic in first term of gradient, comma added at end
as $k\to\infty$. Here $j_l$ is the spherical Bessel function of order
$l$ and $A_{kl}$ is a nonzero constant.
\end{lemma}

In the following lemma we construct an even complex-valued solution of the Helmholtz
equation:
\[
\De\vp+\vp=0,%Edit: colon and comma added
\]
in $\RR^3$ such that the link $L$ is a union of connected components
of its nodal set up to a diffeomorphism. The function~$\vp$ is smooth and
is conveniently given by a finite sum of spherical Bessel functions
and spherical harmonics. We observe that $\vp$ is not square-integrable but it has an optimal decay rate at infinity among all solutions to the Helmholtz equation. The proof of this lemma is presented in
Section~\ref{S.lemma2}, and exploits ideas introduced in~\cite{Annals,Adv,Acta}:

\begin{lemma}\label{L.finite}
There are finitely many complex numbers $c_{lm}$, with $0\leq l\leq
l_0$ and $-l\leq m\leq l$, such that the complex-valued function
\[
\vp:=\sum_{l=0}^{l_0}\sum_{m=-l}^l c_{lm}\, j_l(r)\, Y_{lm}(\theta,\phi)
\]
has the following properties:
\begin{enumerate}

\item The function $\vp$ is even, so $c_{lm}=0$ for all odd $l$.

\item There is a diffeomorphism $\Phi_1$ of $\RR^3$ such that 
  $\Phi_1(L)$ is a union of connected components of the zero set
  $\vp^{-1}(0)$.

\item $\Phi_1(L)$ is structurally stable. More precisely, let $S$ be a
  compact set containing $\Phi_1(L)$. Then there is some
  $\de>0$ such that for any function $\vp'$ with
  $\|\vp-\vp'\|_{C^1(S)}<\de$ one can find a diffeomorphism $\Phi_2$
  of $\RR^3$ such that $\Phi_2\circ\Phi_1(L)$ is a collection of
  connected components of $\vp'^{-1}(0)$ that are contained in $S$.
\end{enumerate}
\end{lemma}

Let us take a large
integer $\hk$ that will be fixed later, and which we assume to be
larger than $\tfrac{l_0}2$. For each even integer $l$ smaller than $2\hk$ we set %Edit: Changed l_0 to l_0/2
\begin{equation}\label{kl}
\hk_l:=\hk-\frac l2\,,
\end{equation}
so that the eigenvalue 
\begin{equation}\label{lahk}
\la:=\la_{\hk_ll}=4\hk+3
\end{equation}
does not depend on the choice of $l$. The desired eigenfunction $\psi$ of the harmonic oscillator can then be %Edit: Added `then' after `can'
derived from the function~$\vp$ constructed in Lemma~\ref{L.finite} by
setting
\[
\psi:=\sum_{l=0}^{l_0}\sum_{m=-l}^l \frac{c_{lm}}{A_{\hk_l l}}\, \psi_{\hk_l lm}
\] 
for a large enough number $\hk$. Notice that, by construction, $\psi$
is a smooth complex-valued function in $L^2(\RR^3)$ that satisfies the
Eq.~\eqref{harm} with $\la$ as in~\eqref{lahk}. Here we have used that
$c_{lm}=0$ for odd $l$, since the number $\hk_l$ defined in~\eqref{kl}
is an integer only for even~$l$.

Let us fix some $R>0$ such that the ball centered at the origin and of
radius $R$, which we will denote by $B$, contains the link
$\Phi_1(L)$. We claim that for any $\de>0$ one can choose $\hk$ large
enough so that
\begin{equation}\label{C1}
\bigg\|\psi\bigg(\frac\cdot{\sqrt\la}\bigg)-\vp(\cdot)\bigg\|_{C^1(B)}<\de\,.
\end{equation}
This is a rather straightforward consequence of
Lemma~\ref{L.asymp}. Indeed, using this lemma, an elementary computation shows that
\begin{align*}
\psi_{klm}\bigg(\frac x{\sqrt{\la_{kl}}}\bigg) & =A_{kl}\,\bigg[ j_l(r)\,
Y_{lm}(\te,\phi)+O\bigg(\frac1k\bigg)\bigg]\,,\\[1mm]
\nabla_x\psi_{klm}\bigg(\frac x{\sqrt{\la_{kl}}}\bigg) & = A_{kl}\,\bigg[
j_l'(r)\, Y_{lm}(\te,\phi)\, e_r+ \frac{j_l(r)}r\, \nabla_{\SS^2}Y_{lm}(\te,\phi)+O(\tfrac1k)\bigg]\,.
\end{align*}%Edit: Added e_r in relevant spot for gradient formula
Hence, substituting these asymptotic expressions in the sum for~$\psi$
we find
\begin{align*}
\bigg|\psi\bigg(\frac x{\sqrt\la}\bigg)-\vp(x)\bigg| 
&\leq
\sum_{l=0}^{l_0}\sum_{m=-l}^l {c_{lm}}\,
\bigg|\frac1{A_{\hk_l l}}\psi_{\hk_l lm}\bigg(\frac x{\sqrt\la}\bigg) -
j_l(r)\, Y_{lm}(\theta,\phi)\bigg|\\
&=\sum_{l=0}^{l_0}\sum_{m=-l}^l {c_{lm}} O(\tfrac1{\hk_l})\leq \frac
C{\hk-\tfrac{l_0}2} \leq \frac C{\hk}
\end{align*}%Edit: l_0 changed to l_0/2 in second last inequality
provided $\hk$ is much larger than $\tfrac{l_0}2$ and $|x|<R$. An analogous %Edit: l_0 changed to l_0/2
argument shows
\[
\bigg|\nabla_x\psi\bigg(\frac
x{\sqrt\la}\bigg)-\nabla_x\vp(x)\bigg|\leq\frac C{\hk}\,,
\]
so the estimate~\eqref{C1} follows provided $\hk$ is large enough.

Item~(iii) in Lemma~\ref{L.finite} ensures that, if $\de$ is small enough, the
function $\psi(\cdot/\sqrt\la)$ has a collection of connected
components in its nodal set $\{\psi(\cdot/\sqrt\la)=0\}$ given by the
link $\Phi_2\circ \Phi_1(L)$, where $\Phi_2$ is a diffeomorphism of
$\RR^3$ and $\Phi_2\circ \Phi_1(L)$ is contained in $B$. Item~(iii)
also ensures that the link $\Phi_2\circ \Phi_1(L)$ is structurally stable for the eigenfunction $\psi$. This implies
that the rescaled link $\Phi_3\circ\Phi_2\circ\Phi_1(L)$ is a union of
structurally stable connected components of $\psi^{-1}(0)$, where $\Phi_3$ denotes the
diffeomorphism of $\RR^3$ given by the rescaling
\[
\Phi_3(x):=\frac{x}{\sqrt\la}\,.
\]
The theorem then follows by setting $\Phi:= \Phi_3\circ\Phi_2\circ\Phi_1$.

\begin{remark}
It is worth noting that the fact that the function~$\vp$ is even was 
key to constructing the radial quantum number $k_l$ via Eq.~\eqref{kl}. A 
straightforward modification of the argument enables us to
consider the case where $\vp$ is odd. As is well known, all
eigenfunctions of the harmonic oscillator must have a definite parity.
In particular the nodal set of the eigenfunction~$\psi$ contains
(at least) two copies of the link $\Phi(L)$ as the link $\Phi_1(L)$
constructed in the proof of Lemma~\ref{L.finite} is contained in the
positive octant of $\RR^3$, which implies that so is $\Phi(L)$. Moreover, $\Phi(L)$ is contained in a small ball of radius $R\la^{-1/2}$.  
\end{remark}

\section{Proof of Lemma~\ref{L.asymp}}
\label{S.lemma1}

The lemma essentially follows from Hilb's asymptotic formula for the Laguerre
polynomials~\cite[Theorem 8.22.4]{Szego75}:
\begin{align*}%\label{LagAsymp}
 \e^{-\frac{r^2}2}r^l L_k^{l+\frac12}(r^2)&=A_{kl}\,
 j_l(\sqrt{\la_{kl}} r)+ O(k^{\frac {l-1}2})\,,\\
\frac{d}{dr}\big[\e^{-\frac{r^2}2}r^l L_k^{l+\frac12}(r^2)\big]&=\sqrt{\la_{kl}} \,\big[A_{kl}\,
 j_l'(\sqrt{\la_{kl}} r)+ O(k^{\frac {l-1}2})\big]\,,
\end{align*}
with
\begin{align*}
A_{kl}&:=
\frac2{\sqrt\pi}\bigg(\frac{\sqrt{\la_{kl}} }2\bigg)^{-l}\frac{\Gamma(k+l+\frac32)}{k!}\,.
\end{align*}
This formula holds uniformly for $r\leq R$. (In fact, the formula for the derivative does not appear in the above
reference, but it is standard ---and easy to prove--- that this
asymptotic formula can be derived term by term). 

The asymptotic expansion
for~$\psi_{klm}$ written in the lemma follows from the
identity~\eqref{psiklm} and the fact that the constant $A_{kl}$ can be
estimated for large~$k$ as
\begin{align*}
A_{kl}&=\frac2{\sqrt\pi}k^{\frac{l+1}2}+ O(k^{\frac{l-1}2})\,.
\end{align*}
This is an elementary computation using Stirling's formula for the
factorial and the identity
\[
\Gamma(k+l+\tfrac32)=\frac{\sqrt\pi\,  (2k+2l+2)!}{2^{2k+2l+2}\,(k+l+1)!}\,.
\]

\section{Proof of Lemma~\ref{L.finite}}
\label{S.lemma2}

Let $B$ be a ball centered at the origin that contains the link $L$. There is no loss of generality in assuming that $L$ is contained in the positive octant of $\RR^3$, that is,
$$L\subset B\cap \{x_1>0, x_2>0, x_3>0\}\,.$$
An easy application of Whitney's approximation theorem ensures that, by perturbing the link a little if necessary, we can assume that it is a real analytic submanifold of $\RR^3$.

Let us denote by $L_\al$ the connected components of $L$, with the
index $\al$ taking values in a finite set $A$. Each component
$L_\al$ is an analytic closed curve without self-intersections. Our
next goal is to write the curve $L_\al$ as the transverse intersection
of two surfaces $\Si_\al^1$ and $\Si_\al^2$.

Since
any closed curve in $\RR^3$ has
trivial normal bundle~\cite{Ma59}, there exists an analytic
submersion $\Theta_\al:W_\al\to\RR^2$, where $W_\al$ is a tubular
neighborhood of $L_\al$ and $\Theta_\al^{-1}(0)=L_\al$. We can then take the
analytic surfaces
$\Si^1_\al:=\Theta_\al^{-1}((-1,1)\times\{0\})\subset W_\al$ and
$\Si^2_\al:=\Theta_\al^{-1}(\{0\}\times(-1,1))\subset W_\al$. Since
$\Theta_\al$ is a submersion, these surfaces
intersect transversally at $L_\al=\Si^1_\al\cap\Si^2_\al$.

Now that we have expressed the component $L_\al$ as the intersection
of two real analytic surfaces $\Si^1_{a}$ and $\Si^2_\al$, we can
consider the following Cauchy problems, with $j=1,2$:
\begin{equation*}%\label{Cauchy}
\Delta u_\al^j+u_\al^j=0\,,\qquad u_\al^j|_{\Si^j_\al}=0\,,\qquad \pd_\nu u_\al^j|_{\Si^j_{\al}}=1\,.
\end{equation*}
Here $\pd_\nu$ denotes a normal derivative at the corresponding
surface. The Cauchy--Kowalewski theorem then grants the existence of
solutions $u_\al^j$ to this Cauchy problem in the closure of small
neighborhoods $U^j_{\al}$ of each surface $\Si^j_{\al}$. We can safely
assume that the tubular neighborhoods $U^1_\al\cap U^2_\al$ are small
enough so that the neighborhoods corresponding to distinct components
are disjoint. Now we take the union of these pairwise disjoint tubular neighborhoods, 
\[
U:=\bigcup_{\al\in A}(U^1_\al\cap U^2_\al)\,,
\]
and define a
complex-valued function $\hat\varphi$ on the set $U$ as
$$\hat\varphi|_{U^1_{\al}\cap U^2_\al}:=u_\al^1+iu_\al^2\,.$$

The following properties of $\hat\varphi$ are clear from the construction:
\begin{enumerate}
\item $\hat\varphi$ satisfies the equation
$$\Delta\hat \vp+\hat\varphi=0$$
in the tubular neighborhood $U$ of the link $L$. We can assume without
loss of generality that $U$ is contained in the positive octant $B\cap
\{x_1>0, x_2>0, x_3>0\}$, as is the link $L$.
\item $U$ can be taken small enough so that the nodal set of $\hat\varphi$ is precisely $L$, i.e., $L=\hat\varphi^{-1}(0)$.
\item The intersection of the zero sets of the real and imaginary parts of $\hat\varphi$ on $L$ is transverse, i.e.,
\begin{equation}\label{trans}
\rank(\nabla \Real \hat\varphi(x),\nabla \Imag \hat\varphi(x))=2
\end{equation}
for all $x\in L$.
\end{enumerate}

Let us agree to say that a subset of $\RR^3$ is {\em symmetric}\/ if
it is invariant under the inversion $x\mapsto -x$, and denote by $-U$
the image of the set $U$ under this map. Since $U$ is contained in the
positive octant, $U\cap -U=\emptyset$. 

Let us then define an even
function~$\vp'$ in the symmetric set
\[
U':=U\cup -U,%Edit: comma added
\]
as
\[
\varphi'(x):=\begin{cases}
\hat\varphi(x)& \text{if }x\in U\,,\\
\hat\varphi(-x) & \text{if }x\in - U\,.
\end{cases}
\] 
By construction, $\vp'$ satisfies the Helmholtz equation 
\begin{equation}\label{eqrol}
\Delta\vp'+\varphi'=0
\end{equation} 
in $U'$ and its nodal set consists of $L$ and its mirror image under the inversion $x\mapsto-x$. 

Denote by $S$ a symmetric closed subset of $U'$ whose interior
contains the link $L$. Our next goal is to construct a solution of the
Helmholtz equation in $\RR^3$ that approximates the local solution $\varphi'$ in the set $S$. To this end, let us take a smooth even function $\chi:\RR^3\to\RR$ equal to $1$ in a neighborhood of $S$ and identically zero outside $U'$, and define
a smooth extension $\vp_0$ of the function $\varphi'$ to $\RR^3$ by
setting $\vp_0:=\chi \varphi'$, which is an even function too. Denote by
\[
G(x):=\frac{\cos|x|}{4\pi|x|}
\]
the Green's function of the operator $\De+1$ in $\RR^3$, which
satisfies the distributional equation
\[
\De G+G=-\de_0
\]
with $\de_0$ the Dirac measure supported at $0$. Since $\vp_0$ is
compactly supported, we obviously have
\begin{equation}\label{intbv}
\vp_0(x)=\int_{\RR^3} G(x-x')\, \rho(x')\, dx'
\end{equation}
with $\rho:=-\De\vp_0-\vp_0$. The complex-valued function $\rho$ is
even and its support is contained in the set $U'\backslash S$. Therefore, an easy continuity argument ensures that one
can approximate the integral~\eqref{intbv} uniformly in the compact set~$S$ by a finite
Riemann sum of the form
\begin{equation}\label{vp1}
\varphi_1(x):=\sum_{j=-J}^J \rho_j\, G(x-x_j)\,.
\end{equation}
Specifically, for any $\de>0$ there is a large integer $J$, complex
numbers $\rho_j$ and points $x_j\in U'\backslash S$ such that the
finite sum~\eqref{vp1} satisfies
\begin{equation}\label{est1}
\|\varphi_1-\varphi'\|_{C^0(S)}<\de\,.
\end{equation}
By the symmetry of the integrand, these quantities can be chosen such
that $\rho_0=0$, $\rho_{-j}=\rho_j$ and $x_{-j}=-x_j$ for $j>0$, thus
guaranteeing that $\varphi_1$ is an even function. Here we have used
that $\vp_0=\vp'$ in $S$.

In the following lemma we show how to ``sweep'' the singularities of
the function $\varphi_1$ in order to approximate it in the set $S$ by another function $\varphi_2$ whose singularities are contained in the complement of the ball $B$. The proof is based on a duality argument and the
Hahn--Banach theorem.

\begin{lemma}\label{L.approx}
For any $\de>0$, there is a finite set of points
$\{z_j\}_{j=-J'}^{J'}$ in $\RR^3\backslash \overline B$ and complex numbers $c_j$ such that the
finite linear combination
\begin{equation}\label{eqmwx}
\varphi_2(x):=\sum_{j=-J'}^{J'} c_j\, G(x-z_j)
\end{equation}
approximates the function $\varphi_1$ uniformly in $S$:%Edit: `as' replaced by colon
\begin{equation}\label{GtG}
\|\varphi_2-\varphi_1\|_{C^0(S)}<\de\,.
\end{equation}
Moreover,
\begin{equation*}\label{even}
c_0=0\,,\qquad z_{-j}=-z_j\,, \qquad c_{-j}=c_j, %Edit: comma added to end
\end{equation*} 
for all $j>0$, so that $\varphi_2$ is an even function.
\end{lemma}
\begin{proof}
Consider the space $\cU$ of all complex-valued functions that are linear
combinations of the form~\eqref{eqmwx}, not necessarily even, where
$z_j$ can be any point in $\RR^3\backslash \overline B$ and the %Edit: Space added after $z_j$
constants $c_j$ take arbitrary complex values. Restricting these functions to the set $S$, $\cU$ can
be regarded as a subspace of the Banach space $C^0(S)$ of continuous
complex-valued functions on $S$.

By the Riesz--Markov theorem, the dual of $C^0(S)$ is  the space
$\cM(S)$ of the finite complex-valued Borel measures on $\RR^3$ whose support
is contained in the set~$S$. Let us take any measure
$\mu\in\cM(S)$ such that $\int_{\RR^3} fd\mu=0$ for all
$f\in \cU$. Let us now define a complex-valued function $F\in L^1\loc(\RR^3)$ as
\[
F(x):=\int_{\RR^3} G(\tilde x - x)\,d\mu(\tilde x)\,, %Edit: Order of x and \tilde x switched
\]
so that $F$ satisfies the equation 
\[
\De F+F=-\mu\,.
\]
Notice that $F$ is identically zero on $\RR^3\backslash \overline B$ by the definition of the
measure~$\mu$ and that $F$
satisfies the elliptic
equation
\[
\De F+F=0
\]
in $\RR^3\minus S$, so $F$ is analytic in this set. Hence, since
$\RR^3\minus S$ is connected and contains the set $\RR^3\minus B$, by
analyticity the function $F$ must vanish on the complement of $S$. It
then follows that the measure $\mu$ also annihilates any complex-valued function of the
form $\rho_j\,G(x-x_j)$
because, as the points $x_j$ do not belong to $S$,
\[
0=\rho_jF(x_j)=\int_{\RR^3} \rho_jG(x-x_j)\,d\mu(x)\,.
\]
Therefore 
\[
\int_{\RR^3}\varphi_1\,d\mu=0\,,
\]
which implies that $\varphi_1$ can be
uniformly approximated on~$S$ by elements of the subspace $\cU$, due to a %Edit: `as' replaced by `, due to'
consequence of the Hahn--Banach theorem. Accordingly, there is a finite set of points
$\{z_j\}_{j=1}^{J'}$ in $\RR^3\backslash \overline B$ and complex numbers $c_j$ such that the function 
\begin{equation*}
\hat\varphi_2(x):=\sum_{j=1}^{J'} 2c_j\, G(x-z_j)
\end{equation*}
approximates the function $\varphi_1$ uniformly in $S$:%Edit: `as' replaced by colon
\begin{equation*}%\label{GtG}
\|\hat\varphi_2-\varphi_1\|_{C^0(S)}<\de\,.
\end{equation*}
The lemma then follows by setting %Edit: added `by' after `follows'
$$
\varphi_2(x):=\sum_{j=1}^{J'} c_j\, G(x-z_j)+\sum_{j=1}^{J'} c_j\, G(x+z_j)=:\sum_{j=-J'}^{J'} c_j\, G(x-z_j)
$$
where $c_0=0$, $c_{-j}=c_j$ and $z_{-j}=z_j$. Indeed, since $S$ is a symmetric set, we have that, for all $x\in S$,
$$
\varphi_2(x)-\varphi_1(x)=\varphi_2(x)-\frac{\varphi_1(x)+\varphi_1(-x)}{2}=\frac{\hat\varphi_2(x)-\varphi_1(x)}{2}+
\frac{\hat\varphi_2(-x)-\varphi_1(-x)}{2}\,,
$$
which implies the desired estimate
$$
\|\varphi_2-\varphi_1\|_{C^0(S)}\leq \frac12\|\hat\varphi_2-\varphi_1\|_{C^0(S)}+\frac12\|\hat\varphi_2-\varphi_1\|_{C^0(S)}<\delta\,.
$$
Notice that we have used the identity $\varphi_1(x)=\varphi_1(-x)$
\end{proof}

To complete the proof of the lemma, notice that the even complex-valued
function $\varphi_2$ constructed in Lemma~\ref{L.approx} satisfies 
\begin{equation}\label{eqwmi}
\De \varphi_2+\varphi_2=0
\end{equation}
in the ball $B$, whose interior contains $S$. Let us take spherical
coordinates $(r,\te,\vp)$ in the ball $B$. Expanding the function
$\varphi_2$ (with respect to the angular variables) in a series of
spherical harmonics and using Eq.~\eqref{eqwmi}, we immediately obtain
that $\varphi_2$ can be written in the ball as a Fourier--Bessel series of the form
\[
\varphi_2=\sum_{l=0}^\infty\sum_{m=-l}^l c_{lm}\, j_l(r)\, Y_{lm}(\te,\vp)\,.
\]
Since $\varphi_2$ is even, we have that $c_{lm}=0$ for all odd $l$. As
before, $j_l$ denotes a spherical Bessel function. 

Since the above series converges in $L^2(B)$, for any $\de>0$ there is an integer $l_0$
such that the finite sum
\[
\varphi:=\sum_{l=0}^{l_0}\sum_{m=-l}^l  c_{lm}\, j_l(r)\, Y_{lm}(\te,\vp)
\]
approximates the function $\varphi_2$ in an $L^2$ sense:%Edit: `an' added after `in'
\begin{equation}\label{huwL2}
\|\varphi-\varphi_2\|_{L^2(B)}<\de\,.
\end{equation}
By the properties of spherical Bessel functions, the complex-valued function $\varphi$ is smooth in $\RR^3$ and satisfies the equation 
\begin{equation}\label{eqmhu}
\De\varphi+\varphi=0
\end{equation}
in the whole space. 

Given any smaller ball $B'$, properly contained in $B$ and in turn containing
the set $S$, standard elliptic
estimates allow us to pass from the $L^2$ bound~\eqref{huwL2} to a
uniform estimate
\[
\|\varphi-\varphi_2\|_{C^0(B')}<C\delta\,.
\]
From this inequality and the bounds~\eqref{est1} and~\eqref{GtG} we infer
\begin{equation}\label{boundmi}
\|\varphi-\varphi'\|_{C^0(S)}<C\de\,.
\end{equation} 
Moreover, since $\varphi'$ also satisfies the Helmholtz equation in a
neighborhood of
the compact set $S$ (cf.\ Eq.~\eqref{eqrol}), standard elliptic estimates again imply that
the uniform estimate~\eqref{boundmi} can be promoted to the $C^1$
bound
\begin{equation}\label{lastbound}
\|\varphi-\varphi'\|_{C^1(S)}<C\de\,.
\end{equation}

Finally, since the link $L$ is a union of components of the the nodal
set of $\varphi'$  and satisfies the transversality
condition~\eqref{trans}, the estimate~\eqref{lastbound} and a direct
application of Thom's isotopy theorem~\cite[Theorem 20.2]{AR} imply
that there is a diffeomorphism $\Phi_1$ of $\RR^3$ such that
$\Phi_1(L)$ is a union of components of the zero set
$\vp^{-1}(0)$. Moreover, the diffeomorphism $\Phi_1$ is $C^1$-close to
the identity and different from the identity just in a small
neighborhood of $L$, so we can safely assume that $\Phi_1(L)$ is contained in
$B$. The structural stability of the link $\Phi_1(L)$ for the function
$\varphi$ also follows from Thom's isotopy theorem and the fact that
$\vp$ satisfies the transversality condition 
$$
\rank(\nabla \Real \varphi(x),\nabla \Imag \varphi(x))=2
$$ 
for all $x\in\Phi_1(L)$. This last equation is a consequence of the
$C^1$-estimate~\eqref{lastbound}, the fact that the function~$\vp'$
satisfies the transversality estimate~\eqref{trans} by definition, and
the fact that transversality is an open property under $C^1$-small
perturbations. The lemma then follows.

\section{A remark about the higher dimensional counterpart}
\label{S.remark}

Following Berry, we have considered the construction of a complex-valued eigenfunction
(or two real-valued eigenfunctions) of the harmonic oscillator in %Edit: Moved `eigenfunction' to before brackets and added `eigenfunctions' inside brackets
three dimensions with a prescribed nodal set of codimension $2$ (that
is, a link). It is worth mentioning that essentially the same argument
enables us to construct $n$ eigenfunctions of the harmonic oscillator
in $\RR^d$ with a prescribed nodal set of codimension $n$.

However, a technical condition makes the statement
considerably more involved in the general case. This condition is
associated with the requirement that the level set be structurally
stable. In the situation covered by the main theorem, the structural
stability follows from the important equation~\eqref{trans}, which
plays a crucial role in the proof. The higher dimensional analog of that
relation would then be the requirement that
\begin{equation}\label{rank}
\rank(\nabla\psi_1(x),\dots,\nabla\psi_n(x))=n, %Edit: comma added
\end{equation}
for all $x$ in the prescribed codimension-$n$ nodal set in $\RR^d$, where $\psi_1,\dots,\psi_n$
would be real-valued eigenfunctions of the harmonic oscillator. For
this condition to hold, a topological obstruction is that the normal
bundle of the set~$L$ that we want to prescribe in the nodal set must be
trivial. Geometrically, this is equivalent to the assertion that a
small tubular neighborhood of the submanifold~$L$ must be
diffeomorphic to $L\times\RR^n$. 

Hence we are led to the following result. Since a link always has
trivial normal bundle~\cite{Ma59}, the main theorem corresponds
exactly to the case $d=3$ and $n=2$.

\begin{theorem}\label{T.high}
Let $L$ be a finite disjoint union of codimension-$n$ compact
submanifolds of $\RR^d$ with trivial normal bundle,
and $d\geq3$. If $n=1$, we also assume that $L$ is connected. Then for any large %Edit: `with' replaced by `and'
enough eigenvalue $\la$ of the harmonic oscillator in $\RR^d$ there
are $n$ real-valued eigenfunctions $\psi_1,\dots,\psi_n$ with
eigenvalue $\la$ and a diffeomorphism $\Phi$ of $\RR^d$ such that
$\Phi(L)$ is the union of connected components of the joint nodal set
$\psi_1^{-1}(0)\cap\cdots\cap\psi_n^{-1}(0)$. Furthermore, $\Phi(L)$ is structurally stable.
\end{theorem}

We note that when $d\geq3$ the eigenfunctions of the harmonic oscillator are given by
\[
\psi_{klm}=\e^{-\frac{r^2}2}r^lL_k^{l+\frac{d-2}2}(r^2)Y_{lm}(\om)\,,
\]
and the corresponding eigenvalues are $\la_{kl}=4k+2l+d$. Here
$\om:=x/r$ is a point in $\SS^{d-1}$ and
$Y_{lm}$ are the spherical harmonics in $\SS^{d-1}$ with frequency $l(l+d-2)$, with $l\geq0$,
$m\equiv(m_1,\dots,m_{d-2})$ and
\[
|m_1|\leq m_2\leq \cdots\leq m_{d-2}\leq l\,.
\]
Using again Hilb's asymptotic formula for the Laguerre polynomials~\cite[Theorem 8.22.4]{Szego75}, we get the asymptotic expressions of Lemma~\ref{L.asymp} generalized to any $d\geq3$:
\begin{align*}
\psi_{klm}(x)&=A_{kl}^d\, \big[j_l^{d}(\sqrt{\la_{kl}}\, r) +O(k^{-\min(\frac{d+1}4,2)})\big]\, Y_{lm}(\om)\,,\\
\nabla \psi_{klm}(x)&=\sqrt{\la_{kl}}\,
A_{kl}^d\,\big[(j_l^d)'(\sqrt{\la_{kl}}\, r)+ O(k^{-\min(\frac{d+1}4,2)})\big]\, Y_{lm}(\theta,\phi)\, e_r\\&+ A_{kl}^d\,\big[j_l^d(\sqrt{\la_{kl}}\, r)+O(k^{-\min(\frac{d+1}4,2)})\big]\,\frac{\nabla_{\SS^{d-1}}Y_{lm}(\om)}{r}\,,
\end{align*}
where $j_l^d$ denote the hyperspherical Bessel functions, which satisfy the radial Helmholtz equation in $\RR^d$, and the constants $A_{kl}^d$ are given by
\[
A_{kl}^d:=
\frac1{\Gamma(\frac d2)}\bigg(\frac{\sqrt{\la_{kl}} }2\bigg)^{-l}\frac{\Gamma(k+l+\frac d2)}{k!}\,.
\]

The proof of Lemma~\ref{L.finite} goes exactly as in Section~\ref{S.lemma2} when the codimension is $n\geq2$. In this case, we use the Cauchy--Kolaweski theorem with data on
hypersurfaces $\Si_\al^j$ (${1\leq j\leq n}$), intersecting at a given
component $L_\al$ transversally, to define functions $u_a^j$. When $n=1$, however, the construction of the real-valued 
function $u$ (analogous to the functions $u_\al^j$ considered in
Section~\ref{S.lemma2}, where we are dropping the sub- and superscripts because
$L$ is now connected and $n=1$) cannot be performed using the Cauchy--Kowalweski
theorem because otherwise $u$ would be defined in a small neighborhood
$U$ of $L$. As the complement of $U$ would not be connected, the proof
of Lemma~\ref{L.approx} does not carry over to this case. To
circumvent this difficulty, we take $U$ to be the precompact domain
bounded by $L$ and define $u$ in $s\, U$ as the first Dirichlet
eigenfunction of the Laplacian in this domain, where $s$ is a scale
factor chosen so that the first Dirichlet eigenvalue of $s\, U$
is~1. The details are as in~\cite[Appendix~A]{Adv}. The rest of the proof remains essentially unchanged.

In particular, in three dimensions the general result yields not only Theorem~\ref{T.main}
but also the following

\begin{corollary}
Let $L$ be a compact surface in $\RR^3$. Then
for any large enough eigenvalue $\la$ of the harmonic oscillator in $\RR^3$
there is a real-valued eigenfunction $\psi$ of the harmonic oscillator %Edit: 'complex' changed to `real'
with energy~$\la$ and a diffeomorphism $\Phi$ of $\RR^3$ such that
$\Phi(L)$ is a connected component of the nodal set
$\psi^{-1}(0)$. Furthermore, $\Phi(L)$ is structurally stable.
\end{corollary}

\section*{Acknowledgments}

The authors are supported by the ERC Starting Grants~633152 (A.E.) and~335079
(D.H.\ and D.P.-S.). This work is supported in part by the
ICMAT--Severo Ochoa grant
SEV-2011-0087.

\bibliographystyle{amsplain}

\begin{thebibliography}{99}\frenchspacing

\bibitem{AR}
R. Abraham, J. Robbin, {\it Transversal Mappings and Flows},
Benjamin, New York, 1967.

%\bibitem{Abramowitz64}
%M. Abramowitz, I.A. Stegun, {\it Handbook of Mathematical Functions with Formulas, Graphs, and
%  Mathematical Tables}, National Bureau of Standards Applied
%  Mathematics Series, Washington D.C., 1964.

%\bibitem{Alessandrini}
%G. Alessandrini, Nodal lines of eigenfunctions of the fixed membrane problem in general convex domains, Comment. Math. Helv. 69 (1994) 142--154.

% \bibitem{Anne}
% C. Ann\'e, Spectre du laplacien et \'ecrasement d'anses,
% Ann. Sci. \'Ecole Norm. Sup. 20 (1987) 271--280.


%\bibitem{BU83}
%S. Bando, H. Urakawa, Generic properties of the eigenvalue of the Laplacian for compact Riemannian manifolds, Tohoku Math. J. 35 (1983) 155--172. 

% \bibitem{B97}
% C. Bar, On nodal sets for Dirac and Laplace operators, Comm. Math. Phys. 188 (1997) 709--721. 

\bibitem{Be01}
M. Berry, Knotted zeros in the quantum states of hydrogen, Found. Phys. 31 (2001) 659--667.

\bibitem{De10}
M.R. Dennis, R.P. King, B. Jack, K. O'Holleran, M.J. Padgett, Isolated optical vortex knots, Nature Phys. 6 (2010) 118--121. 

%\bibitem{Besse08}
%A.L. Besse, Einstein manifolds, Springer-Verlag, Berlin, (2008)

%\bibitem{Cheng}
%S.Y. Cheng, Eigenfunctions and nodal sets, Comment. Math. Helv. 51
%(1976) 43--55.
%
%\bibitem{CM11}
%T.H. Colding, W.P. Minicozzi II, Lower bounds for nodal sets of eigenfunctions, Comm. Math. Phys. 306 (2011) 777--784. 
%
%\bibitem{Colin}
%Y. Colin de Verdi\`ere, Sur la multiplicit\'e de la premi\`ere valeur
%propre non nulle du laplacien, Comment. Math. Helv. 61 (1986)
%254--270.

% \bibitem{DF}
% H. Donnelly, C. Fefferman, Nodal sets of eigenfunctions on Riemannian manifolds, Invent. Math. 93 (1988) 161--183.

%\bibitem{JDG}
%A. Enciso, D. Peralta-Salas, Critical points of Green's functions on
%complete manifolds, J. Differential Geom. 92 (2012) 1--29.

\bibitem{Annals}
A. Enciso, D. Peralta-Salas, Knots and links in steady solutions of the Euler equation, Ann. of Math. 175 (2012) 345--367.

\bibitem{Adv}
A. Enciso, D. Peralta-Salas, Submanifolds that are level sets of solutions to a second-order elliptic PDE, Adv. Math. 249 (2013) 204--249.

\bibitem{Acta}
A. Enciso, D. Peralta-Salas, Existence of knotted vortex tubes in steady Euler flows, Acta Math. 214 (2015) 61--134.

%\bibitem{EPpre}
%A. Enciso, D. Peralta-Salas, Eigenfunctions with prescribed nodal sets, Preprint, in press J. Differential geom. (2014).

\bibitem{EJN}
A. Eremenko, D. Jakobson, N. Nadirashvili, On nodal sets and nodal domains on $\SS^2$ and $\RR^2$,
Ann. Inst. Fourier (Grenoble) 57 (2007) 2345--2360. 

%\bibitem{Freitas}
%P. Freitas, Closed nodal lines and interior hot spots of the second
%eigenfunction of the Laplacian on surfaces, Indiana Univ. Math. J. 51
%(2002) 305--316.

%\bibitem{Halkin}
%H. Halkin, Implicit functions and optimization problems without continuous differentiability of the data, SIAM J. Control 12 (1974) 229--236.

% \bibitem{HS}
% R. Hardt, L. Simon, Nodal sets for solutions of elliptic equations,
% J. Differential Geom. 30 (1989) 505--522.


%\bibitem{Jak}
%D. Jakobson, D. Mangoubi, Tubular neighborhoods of nodal sets and diophantine approximation,
%Amer. J. Math. 131 (2009) 1109--1135.
%
%\bibitem{JN}
%D. Jakobson, N. Nadirashvili, Eigenfunctions with few critical
%points, J. Differential Geom. 53 (1999) 177--182.

\bibitem{survey}
D. Jakobson, N. Nadirashvili, J. Toth, Geometric properties of eigenfunctions, Russian Math. Surveys 56 (2001) 1085--1105. 

% \bibitem{JS07}
% D. Jakobson, A. Strohmaier, High energy limits of Laplace-type and
% Dirac-type eigenfunctions and frame flows, Comm. Math. Phys. 270
% (2007) 813--833.


\bibitem{Irvine}
D. Kleckner, W.T.M. Irvine, Creation and dynamics of
  knotted vortices, {\em Nature Phys.} 9 (2013) 253--258. 

%\bibitem{Komen}
%R. Komendarczyk, On the contact geometry of nodal sets, Trans. Amer. Math. Soc. 358 (2006) 2399--2413.
%
%\bibitem{Lisi}
%S.T. Lisi, Dividing sets as nodal sets of an eigenfunction of the Laplacian, Algebr. Geom. Topol. 11 (2011) 1435--1443.
%
%\bibitem{Melas}
%A.D. Melas, On the nodal line of the second eigenfunction of the Laplacian in $\RR^2$, J. Differential Geom. 35 (1992) 255--263.

%\bibitem{Reed784}
%M. Reed, B. Simon, Methods of modern mathematical physics IV: Analysis of operators, Academic Press Inc, San Diego (1978)
%
%\bibitem{Uh76}
%  K. Uhlenbeck, Generic properties of eigenfunctions, Amer. J. Math. 98 (1976) 1059--1078.


%\bibitem{Ya82}
%S.T. Yau, Problem section, Seminar on Differential Geometry, Annals of Mathematics Studies 102 (1982) 669--706.
%
%\bibitem{Ya93}
%S.T. Yau, Open problems in geometry, Proc. Sympos. Pure Math. 54, pp. 1--28, Amer. Math. Soc., Providence, 1993. 

\bibitem{KL04}
L.H. Kauffman, S.J. Lomonaco, Quantum knots, Proc. SPIE 5436 (2004) 268--284.

\bibitem{Ma59}
W.S. Massey, On the normal bundle of a sphere imbedded in Euclidean space, Proc. Amer. Math. Soc. 10 (1959) 959--964.


%\bibitem{Arfken05}
%G.~B. Arfken and H.~J. Weber.
%\newblock {\em Mathematical methods for physicists}.
%\newblock Elsevier/Academic Press, Burlington, MA, sixth edition, 2005.

%\bibitem{CanzaniPre}
%Y.~Canzani and P.~Sarnak.
%\newblock On the topology of the zero sets of monochromatic random waves.
%\newblock arXiv:1412.4437v1, 2014.

\bibitem{Szego75}
G. Szeg{\H{o}}, {\em Orthogonal Polynomials}, AMS, Providence, 1975.

\end{thebibliography}

\end{document}